\documentclass[journal]{IEEEtran}

\usepackage{graphicx}
\usepackage{bm}
\usepackage{enumerate}
\usepackage{float}
\usepackage{multicol}
\usepackage{color}
\usepackage{url}
\usepackage{cite}

\usepackage{amsmath}
\usepackage{amsthm}
\usepackage{amssymb}
\DeclareMathOperator*{\argmax}{argmax}

\newtheorem{theorem}{Theorem}
\usepackage{algorithmic}
\usepackage{algorithm}
\usepackage{acronym}
\usepackage{array}
\DeclareMathOperator\erf{erf}

\ifCLASSOPTIONcompsoc
  \usepackage[caption=false,font=normalsize,labelfont=sf,textfont=sf]{subfig}
\else
  \usepackage[caption=false,font=footnotesize]{subfig}
\fi

\acrodef{BS}{base station}
\acrodef{RA}{receive antenna}
\acrodef{PA}{predictor antenna}
\acrodef{IID}{independent and identically distributed}
\acrodef{CDF}{cumulative distribution function}
\acrodef{PDF}{probability density function}
\acrodef{cu}{channel use}
\acrodef{ACK}{acknowledgement}
\acrodef{NACK}{negative acknowledgement}
\acrodef{SNR}{signal-to-noise ratio}
\acrodef{HARQ}{hybrid automatic repeat request}
\acrodef{CSIT}{channel state information at the transmitter side}
\acrodef{INR}{incremental redundancy}
\acrodef{npcu}{nats per channel use}
\acrodef{bpcu}{bits per channel use}
\acrodef{MIMO}{multiple input multiple output }
\acrodef{MRC}{maximum ratio combing}

\begin{document}
\captionsetup{belowskip=0pt,aboveskip=0pt}

\title{On Delay-limited Average Rate of HARQ-based Predictor Antenna Systems}
\author{Hao~Guo,~\IEEEmembership{Student~Member,~IEEE},
        Behrooz~Makki,~\IEEEmembership{Senior~Member,~IEEE},
        Mohamed-Slim Alouini,~\IEEEmembership{Fellow,~IEEE},
        and Tommy~Svensson,~\IEEEmembership{Senior~Member,~IEEE}
\thanks{H. Guo and T. Svensson are with the Department of Electrical Engineering, Chalmers University of Technology, 41296 Gothenburg, Sweden (email: hao.guo@chalmers.se; tommy.svensson@chalmers.se).}
\thanks{B. Makki is with Ericsson Research, 41756 Gothenburg, Sweden (email: behrooz.makki@ericsson.com).}
\thanks{M.-S. Alouini is with the King Abdullah University of Science and Technology,
Thuwal 23955-6900, Saudi Arabia (e-mail: slim.alouini@kaust.edu.sa).}}

\maketitle

\begin{abstract}
Predictor antenna (PA) system is referred to as a system with two sets of antennas on the roof of a vehicle, where the PAs positioned in the front of the vehicle are used to predict the channel state observed by the receive antennas (RAs) that are aligned behind the PAs.   In this work, we study the effect of spatial mismatch on the accuracy of channel state information estimation, and analyze the delay-constrained average rate of hybrid automatic repeat request (HARQ)-based PA systems. We consider variable-length incremental redundancy (INR) HARQ protocol, and derive a closed-form expression for the maximum average rate  subject to a maximum delay constraint. Moreover, we study the effect of different parameters such as the vehicle speed, the antenna distance and the rate allocation on the system performance. The results indicate that, compared to the open-loop case, the delay-limited average rate is considerably improved with our proposed PA-HARQ scheme. 
\end{abstract}


\section{Introduction}
One of the main use cases of future wireless networks is vehicle communications. In such setups, \ac{CSIT} plays a key role in the system design and optimization schemes. However, due to the mobility of the users, the typical channel estimation methods may be inaccurate as the position of the transmit/receive antennas may change quickly. For this reason, the  \ac{PA} concept has been proposed and evaluated  in, e.g.,  \cite{Sternad2012WCNCWusing,DT2015ITSMmaking,Guo2019WCLrate,guo2020semilinear,guo2020power}. Here, the \ac{PA} setup refers to as a system where two sets of antennas are deployed on the roof of a vehicle, and the \ac{PA}(s) positioned at the front of the vehicle are used to predict the channel observed by the \acp{RA} that are positioned behind the \ac{PA}(s). 

Typical PA systems have two problems, namely, spatial mismatch and spectral efficiency, compared to  a static \ac{MIMO} system with as many antenna elements. Specifically, due to, e.g., speed change or feedback/processing delay, the \ac{RA}s may not reach the same  point as the \ac{PA}  when they receive data, leading to imperfect \ac{CSIT}. This spatial mismatch problem has been studied in \cite{DT2015ITSMmaking,Guo2019WCLrate,guo2020semilinear} where different techniques  are applied to compensate for this effect. Moreover, while the \ac{PA} is mainly used to improve the CSIT, using one of the antennas only for \ac{CSIT} estimation may not be efficient, in terms of spectral efficiency. This becomes important especially when we remember that, with a \ac{PA} setup on top of buses/trains,  the \ac{PA} should serve a large number of users inside the vehicle, e.g., watching videos. In such cases, we have a  delay-limited system in which a certain number of information bits should be transferred within a limited period of time.

From another perspective, \ac{HARQ} is a well-known technique to improve the data transmission reliability and efficiency. The main idea of \ac{HARQ} is to retransmit the message that experienced poor channel conditions, reducing the outage probability \cite{Caire2001throughput,Sesia2004incremental}. As the \ac{PA} system includes the feedback link, i.e., from the \ac{PA} to the \ac{BS},  \ac{HARQ} can be supported by the \ac{PA} structure in high mobility scenarios. That is, the \ac{BS} could potentially adjust the transmit rate/power to the \ac{RA} based on the feedback from the \ac{PA}. In this way, it is expected that such a joint implementation of the PA and the HARQ scheme can improve the system efficiency and reliability.

In this paper, we develop an \ac{HARQ}-based scheme for delay-limited \ac{PA} systems where we need to transfer a certain number of bits within a limited period of time. The goal is to maximize the delay-constrained average rate in the presence of imperfect \ac{CSIT}. We model the spatial mismatch of the PA system as an equivalent non-central Chi-squared channel with known part from the \ac{PA}-\ac{BS} channel and uncertainties from spatial mismatch, and perform rate adaptation analysis. Particularly, using the \ac{INR} protocol and variable-length coding, we derive closed-form expressions for the maximum delay-limited average rate as well as for the optimal rate allocation. Finally, we evaluate the effect of different parameters such as transmission delay and the vehicle speed on the average rate.  

As opposed to \cite{Sternad2012WCNCWusing,DT2015ITSMmaking}, where numerical and test-bed based analysis are presented, in this work we model the \ac{PA} system analytically, and design an \ac{HARQ}-based \ac{PA} scheme. Moreover, our problem setup  is different from the ones in \cite{Guo2019WCLrate,guo2020semilinear,guo2020power} because of the \ac{HARQ} setup and variable-length coding. Finally, our discussions on the average rate of the \ac{HARQ}-assisted \ac{PA} systems with imperfect \ac{CSIT}  have not been presented before. 

The analytical and simulation results show that, while an HARQ-assisted PA system leads to considerable performance improvement, compared to the open-loop case, the average rate is remarkably affected by the spatial mismatch and delay constraint. Moreover, while our proposed scheme provides the same \ac{CSIT} accuracy as typical \ac{PA} setups, it gives the chance to better utilize the \ac{PA} setup and improve the average rate. Finally, for a broad range of vehicle speeds/delay constraints, the optimal average rate scales with \ac{SNR} almost logarithmically.

\section{System Model}
\begin{figure}
\centering
  \includegraphics[width=1.0\columnwidth]{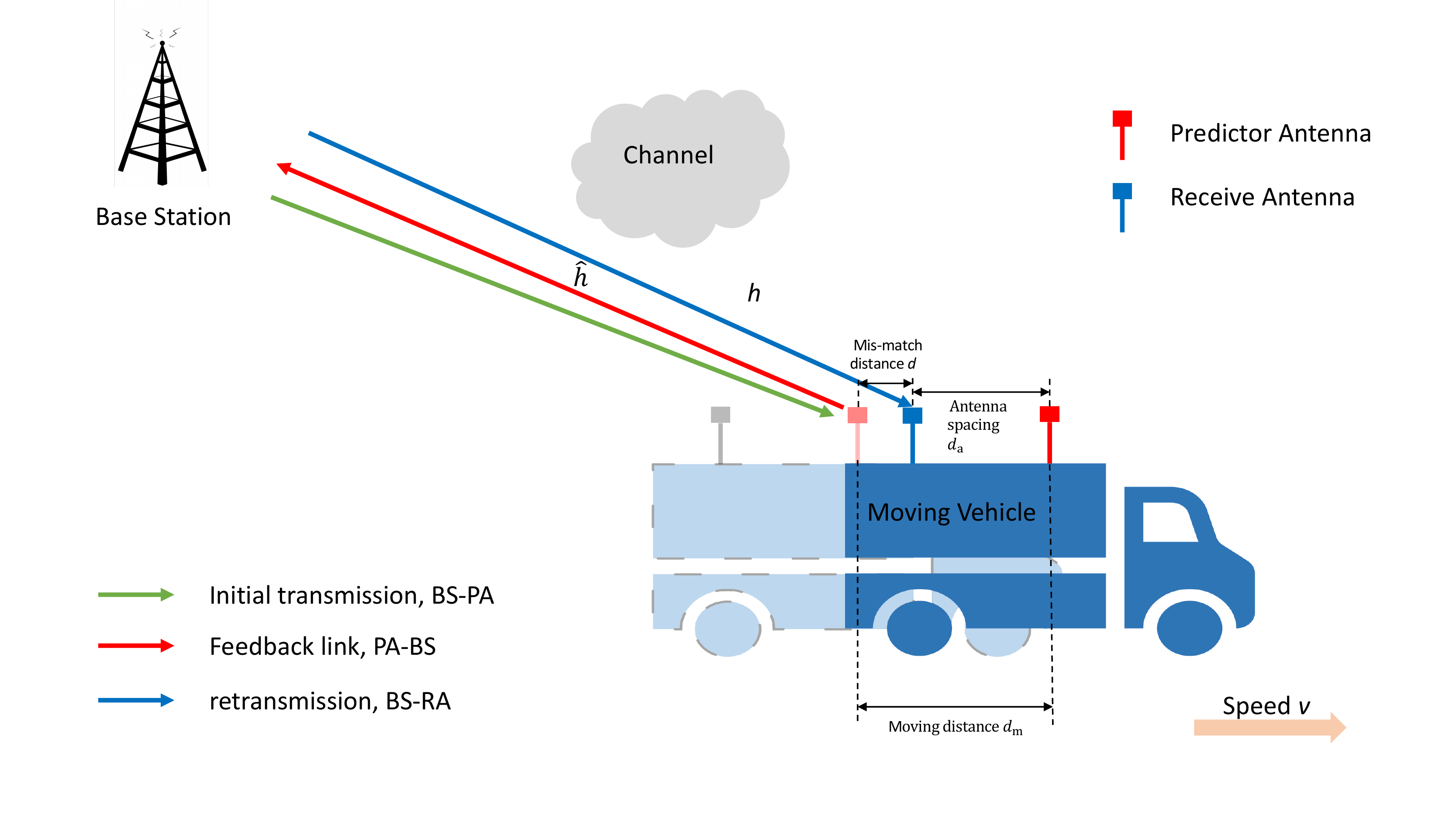}\\
\caption{HARQ-based PA system with spatial mismatch.}
\label{fig_system}
\end{figure}
In the standard PA setup, proposed in \cite{Sternad2012WCNCWusing,DT2015ITSMmaking,Guo2019WCLrate,guo2020semilinear}, the PA at the front of the vehicle\footnote{In the analysis, we consider a forward direction of the vehicle such that the front and the backward antennas work as the PA and the RA, respectively. With a reverse direction, the role of the antennas will be switched.} is used only for channel estimation. As opposed, here, we develop a setup where the PA is used for both channel estimation and initial data transmission under frequency-division duplex (FDD). Our proposed setup and its timing can be seen in Figs. \ref{fig_system} and \ref{fig_timeslot}, respectively. Here,  at $t_1$ the BS sends pilots as well as the encoded data with certain initial rate $R$ and power $p$ to the PA. Then, at $t_2$, the PA estimates the BS-PA channel $\hat{h}$ from the received pilots. At the same time, the PA tries decoding the data and generates an \ac{ACK} if the decoding of the data is successful. Otherwise, a \ac{NACK} will be fed back. \ac{ACK} or \ac{NACK} is sent to the \ac{BS} at $t_3$. In case of an NACK, the \ac{PA} also sends \ac{CSIT} feedback to the BS so that it gets an estimation of $\hat{h}$. Define $t_3-t_1 = \delta$. Based on the received feedback, the \ac{BS}  decides whether or not to retransmit the message with an adapted rate. Finally, at $t_4$ and with \ac{NACK}, the BS sends data to the \ac{RA} with an adapted rate. For simplicity of presentations, we assume the feedback link to be perfect. This is motivated by the fact that, compared to the direct BS-PA link, considerably lower rate is required in the feedback link. Finally, one can follow the same procedure as in \cite{guo2020semilinear} to model the quantization feedback error as an extra additive Gaussian noise.

\begin{figure}
\centering
  \includegraphics[width=1.0\columnwidth]{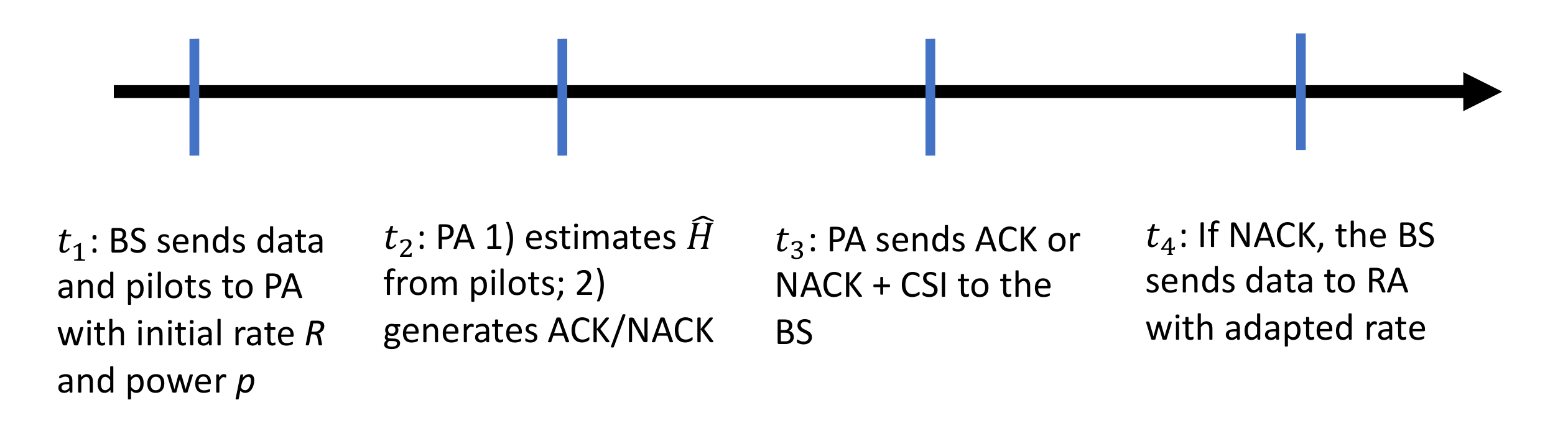}\\
\caption{Timing of  the proposed PA-HARQ scheme.}
\label{fig_timeslot}
\end{figure}

In this way, the signals received by the \ac{PA} and the \ac{RA} in the first and second rounds are
\begin{align}\label{eq_yhat}
\hat{y} = \sqrt{p}\hat{h}\hat{x} + z,
\end{align}
and 
\begin{align}\label{eq_y}
y = \sqrt{p}hx + z,
\end{align}
respectively, where $p$ is the transmit power while $\hat{x}$ and $x$ are the transmitted message with unit variance, and $z \sim \mathcal{CN}(0,1)$ represents  the \ac{IID} complex Gaussian noise added at the receiver side. Also, $h$ represents the BS-RA channel coefficient. Note that, as explained in the following, the signals $\hat{x}$ and $x$ may be of different lengths, as  functions of the channel quality and maximum delay constraint.

We assume that the vehicle moves through a stationary electromagnetic standing wave pattern, and the system is moving in a straight line. Because of the short control-loop delay, the moving distance is in the order of centimeter so the moving direction is almost linear. This assumption has been experimentally verified in, e.g., \cite{Jamaly2014EuCAPanalysis}. Thus, if the \ac{RA} reaches exactly the same position as the position of the \ac{PA} when receiving the pilots, it will experience the same channel, i.e., $h = \hat{h}$, and the \ac{CSIT} will be perfect. However, if the \ac{RA} does not reach the same  point as the \ac{PA}, due to, e.g., the \ac{BS} processing delay is not equal to the time that we need until the \ac{RA} reaches the same  point as the \ac{PA}, the \ac{RA} may receive the data in a place different from the one in which the \ac{PA} was receiving the pilots. Such spatial mismatch may lead to \ac{CSIT} inaccuracy, which  affects the system performance considerably.

Let us define $d$ as the effective distance between  the place where the PA estimates the channel at time $t$, and the place where the RA reaches at time $t+\delta$. Then, we have
\begin{align}\label{eq_d}
    d = |d_\text{a} - d_\text{m} | = |d_\text{a} - v\delta|,
\end{align}
where $d_\text{m}$ is the displacement of the vehicle during time interval $\delta$, and $v$ is the velocity of the vehicle. Then, $d_\text{a}$ is the antenna separation between the PA and the RA. We assume $d$ to be known by the BS. 

Using the classical Jake's correlation model and  assuming a semi-static propagation environment, i.e., assuming that the coherence time of the propagation environment is larger than $\delta$,  the channel coefficient of the BS-RA downlink can be modeled as \cite[eq. (5)]{Guo2019WCLrate}
\begin{align}\label{eq_h}
    h = \sqrt{1-\sigma^2} \hat{h} + \sigma q.
\end{align}
Here, $q \sim \mathcal{CN}(0,1)$ which is independent of the known channel value $\hat{h}\sim \mathcal{CN}(0,1)$, and spatial correlation factor $\sigma$ is a function of the effective distance $d$ as 
\begin{align}\label{eq_sigma}
    \sigma = \frac{\frac{\phi_2^2-\phi_1^2}{\phi_1}}{\sqrt{ \left(\frac{\phi_2}{\phi_1}\right)^2 + \left(\frac{\phi_2^2-\phi_1^2}{\phi_1}\right)^2 }} = \frac{\phi_2^2-\phi_1^2}{\sqrt{ \left(\phi_2\right)^2 + \left(\phi_2^2-\phi_1^2\right)^2 }} .
\end{align}
Furthermore, $\phi_1 = \bm{\Phi}_{1,1}^{1/2}$ and $\phi_2 = \bm{\Phi}_{1,2}^{1/2}$, where $\bm{\Phi}$ is  from Jake's model
\begin{align}\label{eq_tildeH}
     \bigl[ \begin{smallmatrix}
  \hat{h}\\h
\end{smallmatrix} \bigr]= \bm{\Phi}^{1/2} \bm{H}_{\varepsilon},
\end{align}
where $\bm{H}_{\varepsilon}$ has independent circularly-symmetric zero-mean complex Gaussian entries with unit variance, and $\bm{\Phi}$ is the channel correlation matrix with the $(i,j)$-th entry
\begin{align}\label{eq_phi}
    \Phi_{i,j} = J_0\left((i-j)\cdot2\pi d/ \lambda\right) \forall i,j.
\end{align}
Here, $J_n(x) = (\frac{x}{2})^n \sum_{i=0}^{\infty}\frac{(\frac{x}{2})^{2i}(-1)^{i} }{i!\Gamma(n+i+1)}$ represents the $n$-th order Bessel function of the first kind. Moreover, $\lambda$ denotes the carrier wavelength, i.e., $\lambda = c/f_\text{c}$ where $c$ is the speed of light and $f_\text{c}$ is the carrier frequency. Note that if, instead of uniform Jake’s model, we consider different scattering models \cite[Table 6.2]{goldsmith2005wireless}  to study the effect of spatial correlation, (\ref{eq_phi}) is changed to 
\begin{align}\label{eq_phi_gaussina}
    \Phi_{i,j} = e^{-\left(\pi d/\lambda\right)^2} \forall i,j,
\end{align}
for Gaussian scattering, and
\begin{align}\label{eq_phi_rec}
    \Phi_{i,j} = \text{sinc}\left(\frac{2d}{\lambda}\right) \forall i,j,
\end{align}
for rectangular scattering (see Section \ref{Sec:IV} for discussions).

From (\ref{eq_h}), for a given $\hat{h}$ and $\sigma \neq 0$, $|h|$ follows a Rician distribution.  Let us define the channel gains between \ac{BS}-\ac{PA}   and \ac{BS}-\ac{RA}  as $\hat{g} = |\hat{h}|^2$ and $ g = |{h}|^2$, respectively.  Then, the \ac{PDF} of $g|\hat{g}$ is given by
\begin{align}\label{eq_pdf}
    f_{g|\hat{g}}(x) = \frac{1}{\sigma^2}e^{-\frac{x+(1-\sigma^2)\hat{g}}{\sigma^2}}I_0\left(\frac{2\sqrt{x(1-\sigma^2)\hat{g}}}{\sigma^2}\right),
\end{align}
which is non-central Chi-squared distributed with the \ac{CDF}
\begin{align}\label{eq_cdf}
    F_{g|\hat{g}}(x) = 1 - Q_1\left( \sqrt{\frac{2(1-\sigma^2)\hat{g}}{\sigma^2}}, \sqrt{\frac{2x}{\sigma^2}}  \right),
\end{align}
with $Q_1(\cdot,\cdot)$ being the Marcum-Q-function. As shown in Fig. \ref{fig_block}, let us denote the maximum acceptable delay constraint by $L_{\text{max}}$ channel use.  That is, $L_\text{max}$ gives the maximum acceptable delay, or the maximum possible coding length, for data transmission. We present the block length allocated to the first and the second transmission signals by $L$ and $L(\hat g)$, respectively, with $L+L(\hat g)\le L_{\text{max}}$. Here, $L$ is optimized based on average performance, and $L(\hat g)$ is based on the instantaneous channel quality $\hat{g}$. This is because we do not have instantaneous \ac{CSIT} in Round 1, and, therefore, we should send the data with a fixed rate which is determined based on all possible values of channel realization weighted by the channel distribution. However, in the second round, the sub-codeword length is adapted based on the instantaneous channel quality $\hat{g}$ available at the \ac{BS}. In the following, we optimize the average rate in the presence of variable-length coding and a maximum delay constraint.

\begin{figure}
\centering
  \includegraphics[width=1\columnwidth]{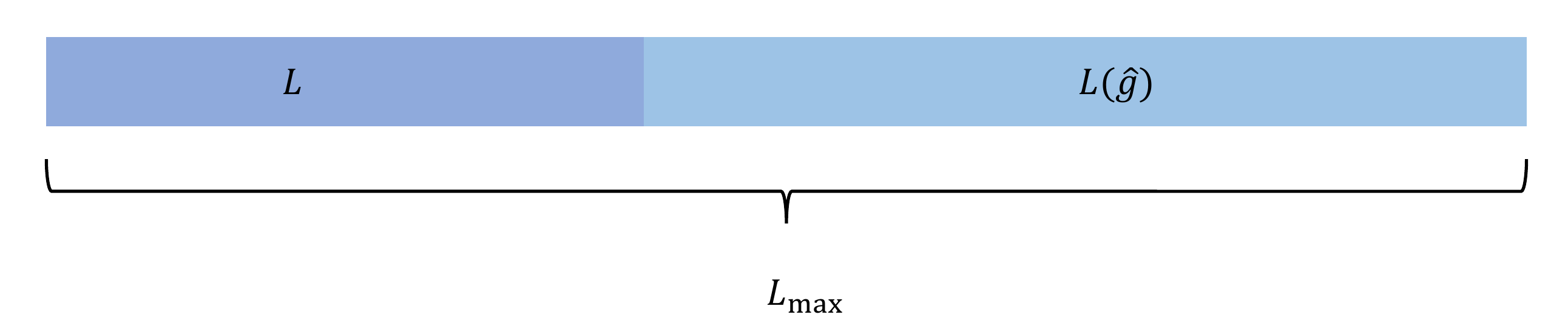}\\
\caption{Illustration of block length for two transmission rounds. }
\label{fig_block}
\end{figure}

\section{Analytical Results}\label{Sec. III}
Let us denote the initial rate by $R = \frac{K}{L}$ \ac{npcu}, where $K$ is the number of nats\footnote{Here, the results are presented in terms of nats, i.e., symbols. It is straightforward to represent the results in terms of bits by scaling the rate terms by $\log_2e$.} per codeword. Here, the results are presented for the cases with moderate/long codewords, which is of interest in moving relays \cite{DT2015ITSMmaking}. In such cases, the error probability converges to the outage probability and we can follow the same information theoretical approach as in \cite{Caire2001throughput,Sesia2004incremental} to analyze the system performance. Then, in the first transmission, if the instantaneous channel gain supports the rate, i.e., $R < \log(1+\hat{g}p)$ or equivalently, $\hat{g} > \frac{\theta}{p}$ where $\theta = e^{R}-1$, an \ac{ACK} is sent back to the \ac{BS} and the transmission stops. On the other hand, if $R > \log(1+\hat{g}p)$, an \ac{NACK} as well as the instantaneous channel gain $\hat{g}$ are sent back to the \ac{BS}.  Then, in Round 2, we determine the required rate as  
\begin{align}\label{eq_rghat}
    R(\hat{g}) = \log(1+\hat{g}p) = \frac{K}{L+L(\hat{g})},
\end{align}
which is equivalent to
\begin{align}\label{eq_lg}
    L(\hat{g}) = \frac{K}{\log(1+\hat{g}p)} - L,
\end{align}
with $L(\hat g)$ being the length of the redundancy bits sent in the second round. Let us denote  $R_\text{min} = \frac{K}{L_\text{max}}$ which is limited by the delay constraint $L_\text{max}$. If $L(\hat{g}) + L > L_{\text{max}}$, with $ \theta_{\text{min}}= e^{R_{\text{min}}}-1$, leading to $\hat{g}<\frac{\theta_{\text{min}}}{p}$, no signal is sent in the second round  because the maximum delay constraint can not be satisfied. Otherwise, if $\hat{g}>\frac{\theta_{\text{min}}}{p}$, we send a sub-codeword with length $L(\hat{g})$. Note that, following the standard INR protocol \cite{Caire2001throughput,Sesia2004incremental}, in the retransmissions only redundancy bits are sent and one can follow the state-of-the-art variable-length coding schemes developed for INR, e.g., \cite{Sesia2004incremental,ma2017incremental}, to generate them.

In this way, the average rate, for given $\hat{g}$, is obtained as
\begin{align}\label{eq_etapre}
    \eta|\hat{g} = R\mathcal{I}\left\{\hat{g}>\frac{\theta}{p}\right\} +\log(1+\hat{g}p)\mathcal{I}\left\{\frac{\theta_{\text{min}}}{p}\leq\hat{g}\leq\frac{\theta}{p}, \mathcal{A}\right\}.
\end{align}
Here, $\mathcal{I}(x)$ is  the indicator  function and $\mathcal{A}$ is the event that the message is decoded in the second round.  Note that, with the rate allocation scheme of (\ref{eq_rghat}), it is straightforward to show that $\mathcal{A}$ occurs if and only if $g \geq \hat{g}$. In this way, (\ref{eq_etapre}) can be rewritten as
\begin{align}\label{eq_eta}
       \eta|\hat{g} = R\mathcal{I}\left\{\hat{g}>\frac{\theta}{p}\right\} + \log(1+\hat{g}p)\mathcal{I}\left\{\frac{\theta_{\text{min}}}{p}\leq\hat{g}\leq\frac{\theta}{p}, g\geq\hat{g}\right\}. 
\end{align}


Finally, averaging  the system performance over all possible $\hat{g}$, we can obtain the average rate as follows.
\begin{theorem}\label{theorem1}
The average rate of the proposed PA-HARQ scheme is approximately given by (\ref{eq_etabar}).
\end{theorem}
\begin{proof}
Using (\ref{eq_eta}) and taking average over $\hat{g}$, whose  \ac{PDF} is given by $f_{\hat{g}}(x) =  e^{- x}$, we have
\begin{align}\label{eq_etabar}
    \eta(R) &= R\Pr\left(\hat{g}\geq\frac{\theta}{p}\right) + \int_{\theta_{\text{min}}/p}^{\theta/p} \log(1+px)f_{\hat{g}}(x)\Pr(g\geq x) \text{d}x\nonumber\\
    &= Re^{-\frac{\theta}{p}}\times\nonumber\\ &~~\int_{\theta_{\text{min}}/p}^{\theta/p}   e^{- x}\log(1+px)\mathcal{Q}\left(\sqrt{\frac{2\left(1-\sigma^2\right)x}{\sigma^2}},\sqrt{\frac{2x}{\sigma^2}}\right) \text{d}x\nonumber\\
    &\overset{(b)}{\simeq} Re^{-\frac{\theta}{p}} + 
    \int_{\theta_{\text{min}}/p}^{\theta/p} e^{-x}\log(1+px)\frac{1}{2} \text{d}x +\nonumber\\ &~~\frac{1}{2}\int_{\theta_{\text{min}}/p}^{\theta/p} e^{-x}\log(1+px)\frac{\sigma}{2\sqrt{\pi x}} \text{d}x \nonumber\\
    &\overset{(c)}{=}Re^{-\frac{\theta}{p}} + \left(F_1\left(\theta/p\right) - F_1\left(\theta_{\text{min}}/p\right)\right) +\nonumber\\
    &\frac{\sigma}{4\sqrt{\pi}}\Bigg(  k\left(F_2(\theta/p) - F_2(\theta_{\text{min}}/p)\right)+\nonumber\\ &~~b\sqrt{\pi}\left(\erf\left(\sqrt{\frac{\theta}{p}}\right)-  \erf\left(\sqrt{\frac{\theta_{\text{min}}}{p}}\right) \right) \Bigg).
\end{align}
Here, $(b)$ is obtained by $\sqrt{\frac{2\left(1-\sigma^2\right)x}{\sigma^2}} \simeq \sqrt{\frac{2x}{\sigma^2}}$ for small/moderate values of $\sigma$, using \cite[Eq.  (A-3-2)]{schwartz1995communication}
\begin{align}
    \mathcal{Q}(x,x) = \frac{1}{2}\left(1+e^{-x^2}I_0(x^2)\right), 
\end{align}
and  $I_0(x) \simeq \frac{e^x}{\sqrt{2\pi x}}$ \cite[Eq. (9.7.1)]{abramowitz1999ia}. Then, in $(c)$ we use linear approximation of $\log(1+px) \simeq kx+b$ at $x = \frac{\theta+\theta_{\text{min}}}{2p}$, with $k = \frac{2p}{\theta+\theta_{\text{min}} +2}$ and $b = -\frac{\theta+\theta_{\text{min}}}{2p}k + \log\left(1+\frac{\theta+\theta_{\text{min}}}{2}\right)$. Moreover, $(c)$ is obtained by  defining $ F_1(x)$  and $ F_2(x)$ as 
\begin{align}
    F_1(x) = -\frac{1}{2}e^{-x}\left(\log(px+1) + E_1\left(\frac{px+1}{p}\right)e^{x+\frac{1}{p}}\right),
\end{align}
and
\begin{align}
    F_2(x) = \frac{1}{2}\sqrt{\pi}\erf\left(\sqrt{x}\right) -\sqrt{x}e^{-x},
\end{align}
with some manipulations. Here,  $E_1(z) = \int_z^\infty \frac{e^{-t}}{t} \text{d}t$ represents the exponential integral, and $\erf(x) = \frac{2}{\sqrt{\pi}}\int_0^{x} e^{-t^2}\text{d}t$ is the error function.
\end{proof}
Using (\ref{eq_etabar}), the optimal average rate can be found as $\eta_{\text{opt}} = \argmax_R~ \eta(R)$.  Setting the derivative of (\ref{eq_etabar}) with respect to $R$ equal to zero, the optimal transmit rate $R$ is found by numerical solution of
\begin{align}
    R_\text{opt} &=\underset{R}{\arg} \Bigg\{ -\frac{1}{p}\left(Re^{R}-p\right)e^{-\frac{\theta}{p}}+\nonumber\\
    &~~~\frac{\sigma}{4\sqrt{\pi}}\Bigg(ke^{R-\frac{\theta}{p}} - k\frac{\left(2e^{R}-p-2\right)e^{R-\frac{\theta}{p}}}{2p^2\frac{\theta}{p}} +2be^{R-\frac{\theta}{p}}\Bigg)+\nonumber\\
    &~~~ \frac{1}{4p}\left(Re^{R}-p\right)e^{-\frac{\theta}{p}}e^{R+\frac{1}{p}}\frac{e^{-\frac{pR+1}{p}}}{\frac{pR+1}{p}} =0\Bigg\},
\end{align}
which is a one-dimensional equation and can be solved effectively by, e.g., bisection method. 

Finally, to evaluate the effectiveness of the proposed setup, we consider following  benchmarks: 
\begin{itemize}
    \item \textit{Open-loop transmission:} The average rate of the open-loop setup, i.e., with no retransmissions, is given by
\begin{align}\label{eq_bench}
    \eta^{\text{Open-loop}} = \int_{\frac{\theta}{p}}^{\infty} Re^{-x} \text{d}x = Re^{-\frac{\theta}{p}},
\end{align}
and the optimal rate allocation is found by setting the derivative of (\ref{eq_bench}) with respect to $R$ equal to zero leading to $R = \mathcal{W}(p) $, where $\mathcal{W}(\cdot)$ represents the Lambert W-function $xe^x = y \Leftrightarrow x = \mathcal{W}(y)$ \cite{corless1996lambertw}.
\item \textit{Basic ARQ:}  Here, we follow the same method as in the proposed scheme, except that in the second round exactly the same signal is retransmitted and the receiver does not combine the two copies of the signal. Then, (\ref{eq_eta}) is rephrased as
\begin{align}\label{eq_eta_basicARQ}
       \eta|\hat{g} = R\mathcal{I}\left\{\hat{g}>\frac{\theta}{p}\right\} + 0.5R\mathcal{I}\left\{\frac{\theta_{\text{min}}}{p}\leq\hat{g}\leq\frac{\theta}{p}, g\geq\hat{g}\right\},
\end{align}
and one can follow the same method as in Theorem \ref{theorem1} to derive the average system performance.
\item \textit{Diversity-based transmission:} Here, we send the same signal in two spectrum resources simultaneously, and do \ac{MRC} at the receiver. Then, the  throughput is given by
\begin{align}\label{eq_eta_diversity}
       \eta|\hat{g},g = 0.5R\mathcal{I}\left\{\log\left(1+\left(g +\hat{g}\right)P\right)>R\right\}. 
\end{align}
\end{itemize}

\section{{Simulation Results}}\label{Sec:IV}

\begin{figure}
\centering
  \includegraphics[width=0.88\columnwidth]{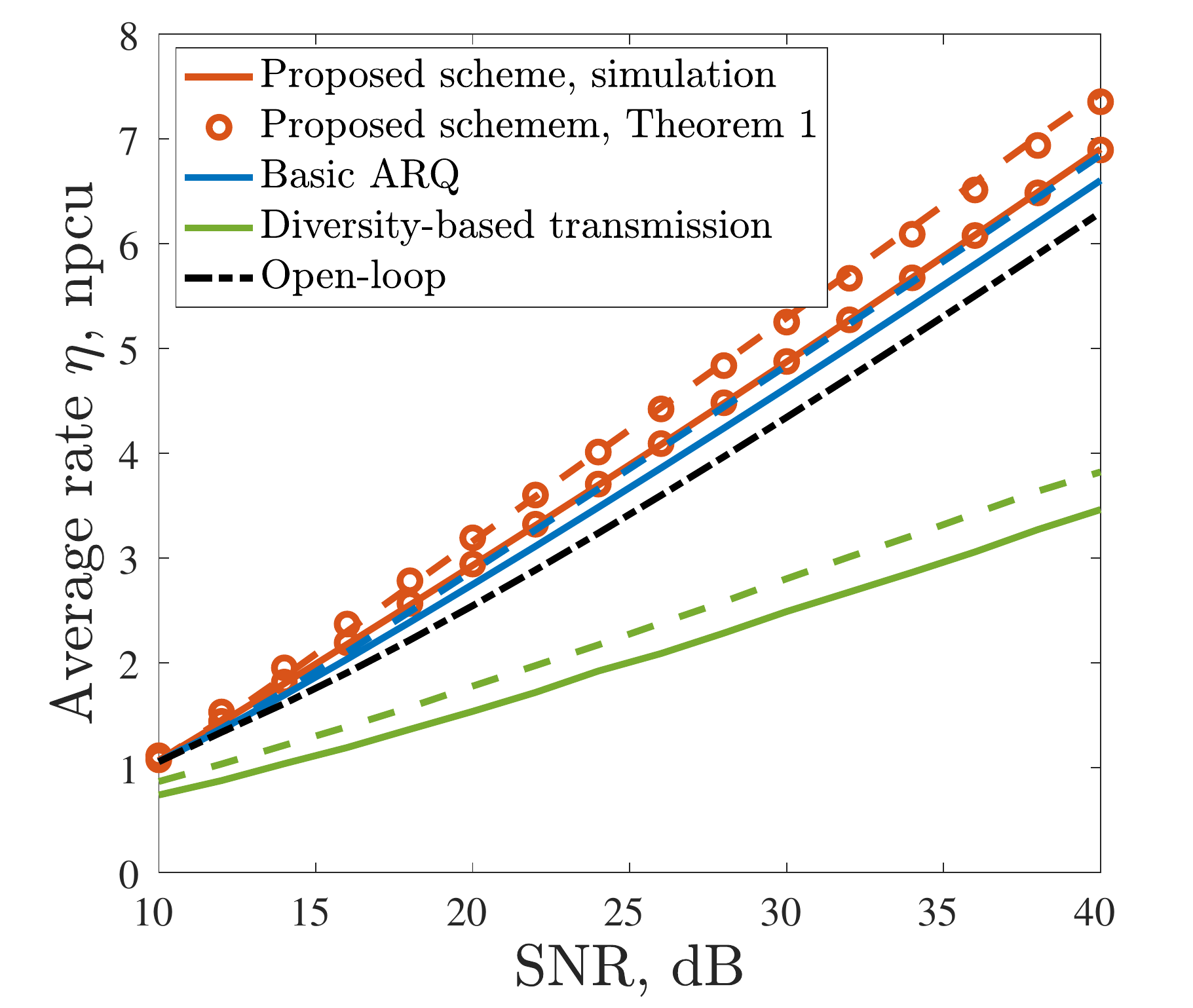}\\
\caption{ $\eta_{\text{opt}}$ vs. SNR. $R_\text{min}$ is set to 2 \ac{npcu}, and $\sigma$ in (\ref{eq_sigma}) is set to 0.1 for dash line and 0.9 for solid line, respectively. }
\label{fig_r1}
\vspace{-5mm}
\end{figure}

In the simulations, we set $ \delta$ = 5 ms, carrier frequency $f_\text{c} = \frac{c}{\lambda}$ = 2.68 GHz with $\lambda$ being the carrier wavelength, and $d_\text{a} = 1.5\lambda$, which affect the  spatial correlation factor $\sigma$ as given in (\ref{eq_sigma}). In Fig. \ref{fig_r1}, we demonstrate the optimal average rate in \ac{npcu} with different cases for a broad range of \ac{SNR}s which, because the noise has unit variance, we define as $10\log_{10}P$. Here, we set $R_\text{min}=2$ \ac{npcu}  which corresponds to a delay constraint $L_\text{max}$ (cf. Sec. III), and $\sigma = 0.1, 0.9$. The approximation values are obtained by Theorem \ref{theorem1}, and the results are compared with the benchmark schemes given in (\ref{eq_bench})-(\ref{eq_eta_diversity}). Figure \ref{fig_r2} presents  optimal values of initial rate allocation $R$ for different values of $R_\text{min}$ as a function of the \ac{SNR}. Here,  we set $R_\text{min}=2,3$ \ac{npcu}, and $\sigma = 0.1$. Then, the optimal average rate as a function of the vehicle velocities is illustrated in Fig. \ref{fig_r3}, with $R_\text{min} = 2$ \ac{npcu}, and \ac{SNR} = 20, 24, and 28 dB. Moreover, we present the approximation values from Theorem \ref{theorem1}.  Here, as explained in \cite[Section II]{Guo2019WCLrate}, manipulating $v$ is equivalent to changing the level of spatial correlation $\sigma$ for given values of $ \delta$, $f_\text{c}$ and $d_\text{a}$ in (\ref{eq_h}). Finally, we study the effect of different scattering models on the system performance, i.e., uniform scattering (\ref{eq_phi}), Gaussian scattering (\ref{eq_phi_gaussina}), and rectangular scattering (\ref{eq_phi_rec}).

\begin{figure}
\centering
  \includegraphics[width=0.88\columnwidth]{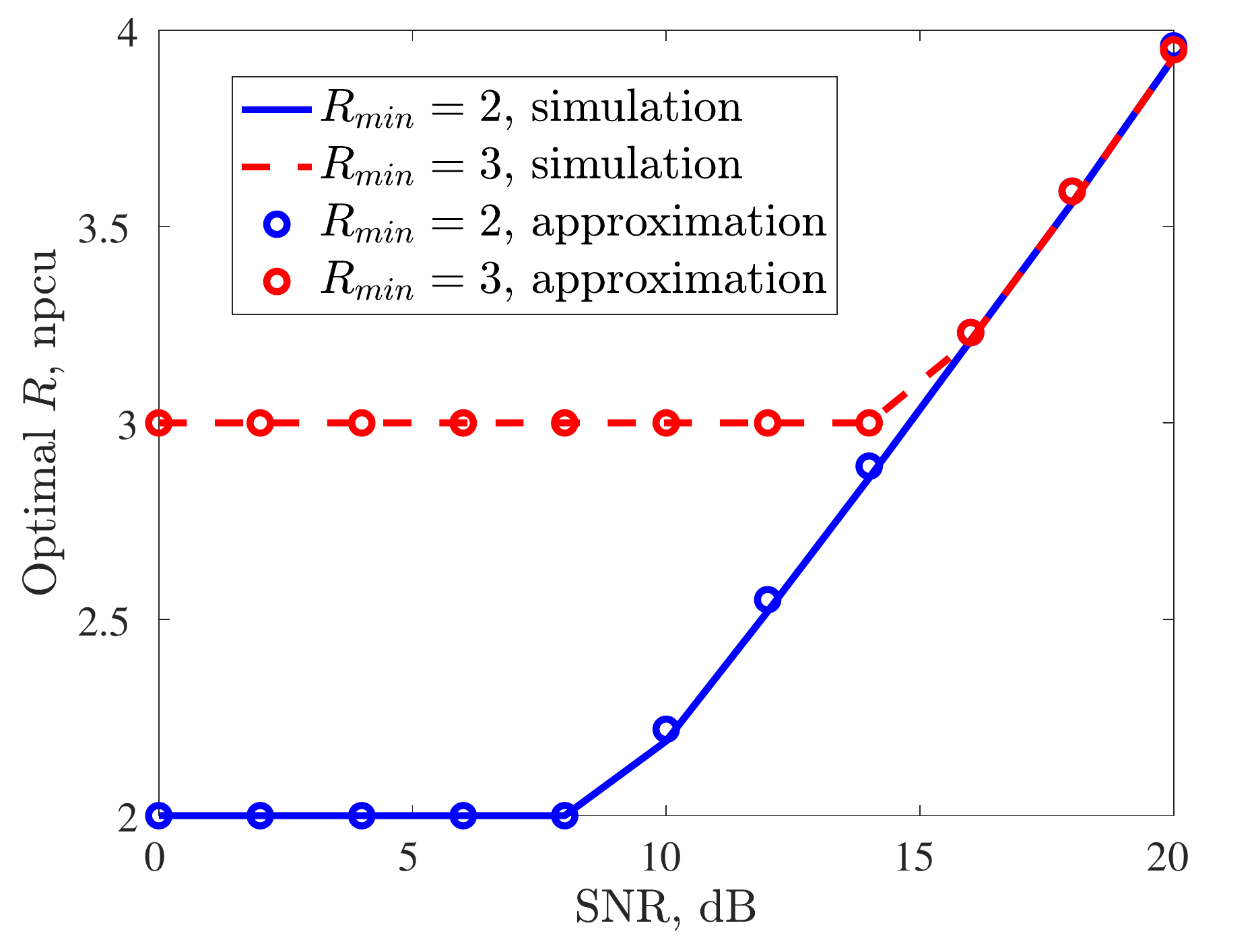}\\
\caption{Optimal values of $R$ vs. SNR, $R_\text{min}=2, 3$ \ac{npcu}, and $\sigma = 0.1$.}
\label{fig_r2}
\vspace{-5mm}
\end{figure}



According to the figures, we can conclude the following:
\begin{itemize}
     \item The approximation scheme of Theorem \ref{theorem1} is  tight for a broad range of values of SNR (Figs. \ref{fig_r1}-\ref{fig_r3}), $R_\text{min}$ (Fig.  \ref{fig_r2}) as well as $\sigma$, or equivalently, $v$ (Figs. \ref{fig_r1}, \ref{fig_r3}). Thus, for different parameter settings, the average rate of the \ac{PA}-\ac{HARQ} system can be well optimized by Theorem \ref{theorem1}.
    \item With deployment of the PA-HARQ setup, remarkable average rate gain  is achieved, compared to benchmark schemes, especially in moderate/high SNRs (Fig. \ref{fig_r1}). Moreover, the average rate increases when the spatial correlation between the \ac{PA} and the \ac{RA} is low, i.e., when $\sigma$ increases (Figs. \ref{fig_r1}, \ref{fig_r3}). 
    \item  For a broad range of parameter settings and SNRs, the optimal rate $R$ increases with \ac{SNR},  in the log domain, almost linearly (Fig. \ref{fig_r2}).  As $R_\text{min}$ increases, i.e., with more stringent delay-constraint, higher  initial transmission rate $R$ is required for the optimal average rate, as can be seen in the low-\ac{SNR} range  of the figure.
    \item  In Fig. \ref{fig_r3}, there is a minimum value for the average rate. This is intuitively because at low and high speeds, i.e., $v<110$ km/h and $v>130$ km/h, we have larger $\sigma$, which although it reduces the CSIT accuracy, provides higher average rate since it gives better spatial diversity. The lowest value of $\eta$ is observed at around 120 km/h, which matches our straightforward calculation from $d = |d_\text{m}-d_\text{a}| $. Moreover, the effect of different scattering models on the system performance is relatively negligible (Fig. \ref{fig_r3}). Thus, the analytical and the simulation-based conclusions hold for different scattering models.
\end{itemize}
\begin{figure}
\centering
  \includegraphics[width=0.88\columnwidth]{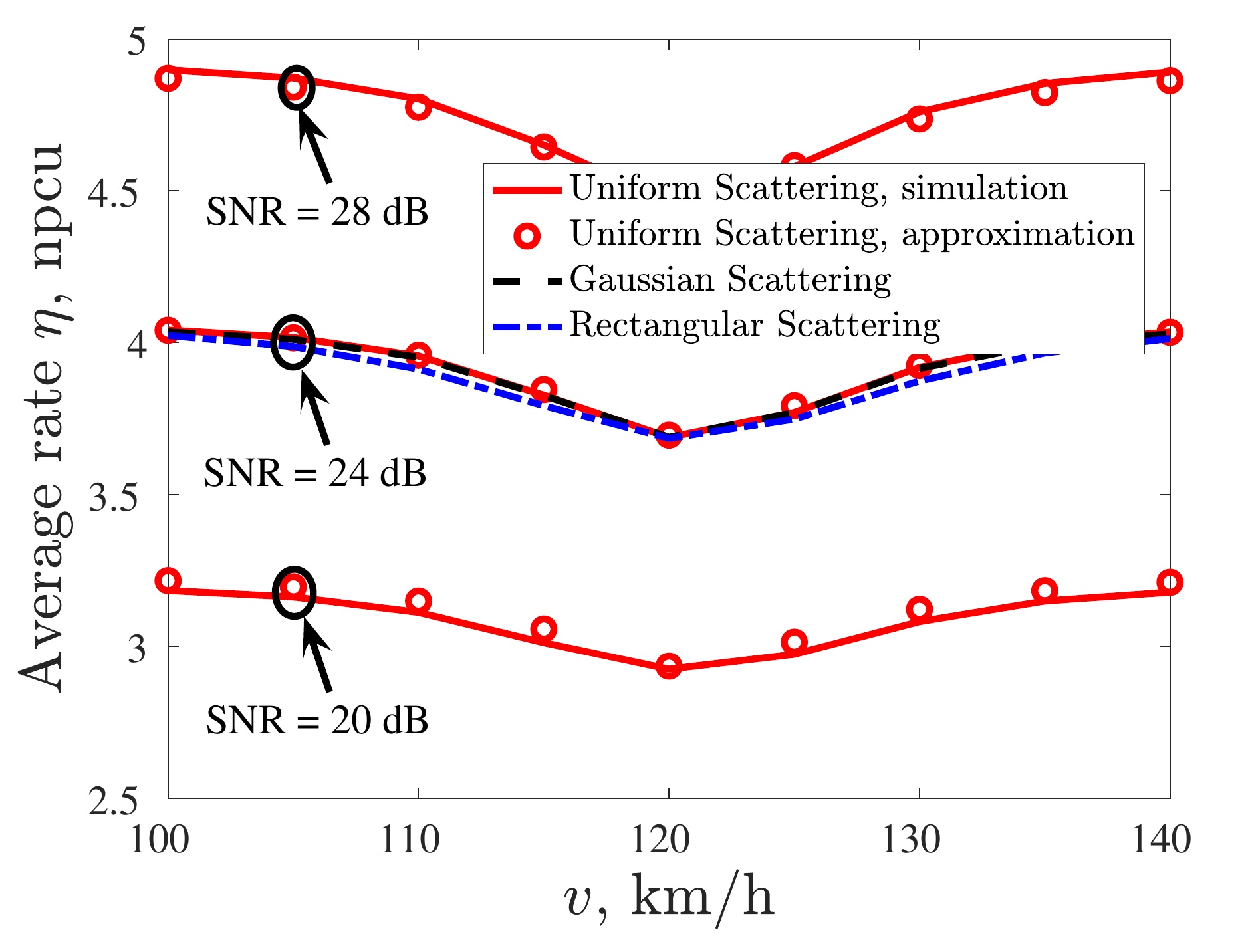}\\
\caption{ $\eta_{\text{opt}}$ vs. $v$. $R_\text{min} = 2$ \ac{npcu},  and SNR = 20, 24, 28 dB.}
\label{fig_r3}
\vspace{-5mm}
\end{figure}
\section{Conclusion}
We proposed an \ac{HARQ}-based approach in order to compensate for the spatial mismatch  and spectral under-utilization  of \ac{PA} systems. We designed the system such that the feedback from the PA can be used to adjust the codeword length in order to improve the average rate. The simulation and analytical results show that, while HARQ-assisted PA systems lead to considerable performance improvement, compared to different benchmark systems, the average rate is remarkably affected by the spatial mismatch, the \ac{SNR} as well as  the delay constraint.
\vspace{-8mm}
\bibliographystyle{IEEEtran}

\bibliography{main.bib}

\end{document}